\documentclass[journal,twocolumn,10pt,twoside]{IEEEtran}

\normalsize
\usepackage[utf8x]{inputenc}
\usepackage{cite}
\usepackage{flushend}

\usepackage[dvips]{graphicx}
\graphicspath{{./}}

\usepackage[cmex10]{amsmath}
\usepackage{array}
\usepackage{multirow}
\usepackage{tabularx}

\usepackage[caption=false,font=footnotesize]{subfig}

\usepackage{amssymb}
\usepackage{steinmetz}
\usepackage{amsthm}
\newtheorem{thm}{Theorem}
\newtheorem{cor}{Corollary}[thm]
\newcommand{\figureWidthScheme}{3.4in}
\newcommand{\figureWidthFigure}{3.4in}
\newcommand{\figureWidthFlowchart}{7cm}

\newcommand{\R}{\mathbb{R}}

\newcommand{\Z}{\mathbb{Z}}

\newcommand{\C}{\mathbb{C}}

\allowdisplaybreaks

\begin{document}

\title{Orthogonal Design of Cyclic Block Filtered Multitone Modulation}
%
%
%

\author{Mauro Girotto,~\IEEEmembership{Member,~IEEE,} and  
        Andrea M. Tonello,~\IEEEmembership{Senior Member,~IEEE}\thanks{A version of this manuscript has been submitted to the IEEE Transactions on Communications for possible publication.\newline
        M. Girotto is with the University of Udine, Udine 33100, Italy (e-mail: mauro.girotto@uniud.it).
       A. Tonello is with the University of Klagenfurt, Klagenfurt 9020, Austria (e-mail: andrea.tonello@aau.at).}}

\maketitle

\begin{abstract}
The orthogonal design of a Cyclic Block Filtered Multititone Modulation (CB-FMT) system is addressed. CB-FMT is a filter bank modulation scheme that uses frequency confined prototype pulses, similarly to Filtered Multitone Modulation (FMT). Differently from FMT, where the linear convolution is used, the cyclic convolution is exploited in CB-FMT. This allows to efficiently implement the system via a concatenation of discrete Fourier transforms (DFT). The necessary and sufficient orthogonality conditions are derived in time domain and frequency domain. Then, these conditions are expressed in matrix form and the prototype pulse coefficients are parameterized with hyper-spherical coordinates. The effect of a linear time-variant transmission medium is discussed. In such a scenario, the optimal filter bank orthogonal design is considered with the objective of maximizing either the in-band-to-out-band sub-channel energy ratio or the achievable rate. Numerical results and comparisons show the performance improvements attainable with several designed optimal pulses also w.r.t. the use of the baseline root-raised-cosine pulse.
\end{abstract}


\IEEEpeerreviewmaketitle



\section{Introduction}
\label{sec:Introduction}
\IEEEPARstart{T}{he large} demand for broadband telecommunications has pushed the research and development of advanced physical layer techniques based on multi-carrier (MC) modulation also known as Filter Bank Modulation (FBM). The idea behind FBM is to partition the wide band frequency selective channel in a number of narrowband sub-channels where the parallel and simultaneous transmission of low data-rate signals is performed. In general, FBM use spectrally well-shaped prototype filters. In this way, the sub-channel frequency responses become nearly flat if the number of sub-channels is sufficiently high. This greatly simplifies the equalization task. In addition, FBM allows to flexibly manage the spectrum occupancy by switching on/off the sub-channels transmission and allocating the available resources through power and bit loading algorithms.    
The most popular FBM scheme is Orthogonal Frequency Division Multiplexing (OFDM) \cite{ref:Bingham}. It has been adopted in several standards, e.g., in the Wireless LAN IEEE 802.11 \cite{WLAN_MAC}, in the wireless MAN IEEE 802.16 \cite{WMAN} and in the new cellular LTE \cite{LTE_com_mag} standards. OFDM is also used in ADSL \cite{ADSL} and it is at the base of both broadband and narrowband power line communication systems \cite{ref:PLC_Book}. OFDM can be viewed as a FBM scheme where the prototype pulse has a rectangular pulse shape in time domain. This pulse shape enables a simple and efficient implementation based on Discrete Fourier Transform (DFT). Despite its simplicity, the poor sub-channel frequency confinement (sinc like) renders OFDM weak in the presence of non-idealities such as channel time-variations introduced by nodes mobility \cite{OFDM_tv_Cimini} and asynchronism between users \cite{ref:Tonello_BLTJ}. The paradigm followed in more general FBM schemes is to deploy frequency confined sub-channel pulses. These architectures are well represented by Filtered Multitone Modulation (FMT) \cite{ref:Cherubini} which synthesizes the transmitter with an exponentially modulated filter bank (FB) with a prototype pulse designed to have high sub-channel frequency confinement. FMT can be implemented with a DFT and polyphase filtering (DFT filter bank) (see \cite{Moret_JWC_Efficient_implementation} and references therein). Typically, long pulses are required to realize good frequency confinement as for instance reported in \cite{Pulse1, Pulse2, Pulse3, Pulse4}. In such design examples, the FB does not have the perfect reconstruction (or orthogonality) property. The realization of orthogonal DFT filter banks has been discussed in \cite{Moret_JWC_Efficient_implementation}, \cite{ORT1}, \cite{Siclet},  where it has been shown to be a complex task. 

In this paper, a different FBM scheme is considered. It is referred to as Cyclic Block Filtered Multitone Modulation (CB-FMT) \cite{CBFMT_ICC}. Similarly to the idea in conventional FMT, the prototype pulse is designed to have high frequency confinement. However, a key difference concerns the filtering operation in the FB: FMT uses the linear convolution while CB-FMT uses the cyclic convolution. This turns the data transmission into a block transmission and the efficient implementation is possible in the frequency domain (FD) via the concatenation of DFTs. This renders the complexity of CB-FMT significantly lower than conventional FMT with equal prototype pulse length \cite{CBFMT_EURASIP}. The orthogonal FB design can be done in the FD as shown in the preliminary results presented in \cite{CBFMT_EW2014}. Also equalization can be implemented in the FD and it has the potentiality of offering better bit-error-rate performance than OFDM in typical wireless fading channels \cite{CBFMT_EURASIP}.

Herein, the design of an orthogonal CB-FMT system is analyzed. A preliminary assessment of this specific problem was reported in \cite{CBFMT_EURASIP,CBFMT_EW2014}. More in detail, in \cite{CBFMT_ICC,CBFMT_EURASIP} the orthogonality conditions were introduced and a very simple design was proposed based on sampling a root-raised-cosine (RRC) pulse spectrum. In \cite{CBFMT_EW2014} the idea of parameterizing the orthogonal relations with angles was introduced. The novel contributions of this paper can be summarized as follows:
\begin{itemize}
\item The orthogonality conditions are analyzed following a rigorous mathematical approach. It is shown that the cyclic convolution allows to write the orthogonality conditions in a simple matrix form. The properties of these matrices are analyzed to enable the constructive design of an orthogonal CB-FMT system.
\item Prototype pulse coefficients are parameterized with hyper-spherical coordinates following the approach introduced in \cite{CBFMT_EW2014}. It is shown that the orthogonality conditions can be reduced into a set of non-linear systems.
\item The effects on orthogonality introduced by a transmission medium that is time variant and dispersive is analyzed. 
\item The optimal filter bank orthogonal design is considered with the objective of maximizing either the in-band-to-out-band sub-channel energy ratio or the achievable rate. Several pulses are then obtained and the system performance is reported. 
\item Starting from an optimal orthogonal mother pulse (designed for a certain set of parameters), a simple method to obtain an orthogonal new pulse in the presence of a parameter variation, e.g., number of sub-channels or pulse length, is reported. This allows to avoid the search of a new pulse.
\end{itemize}

The paper is organized as follows. In Sec.~\ref{Sec:CBFMT}, we briefly recall the CB-FMT principles and the perfect reconstruction conditions are reported in time and frequency domains. In Sec.~\ref{Sec:Orthogonality}, the orthogonality conditions are introduced and they are written in matrix form. The reuse of an orthogonal pulse for a different set of parameters is also considered. In Sec.~\ref{Sec:Practical} the orthogonality problem when transmission takes place in a linear time-variant medium, e.g, in a mobile radio channel, is considered. In Sec.~\ref{Sec:Design}, we introduce a parameterization of the matrix elements in terms of non-linear combination of trigonometric functions and we introduce the objective functions to be maximized in the optimal pulse search. In Sec.~\ref{Sec:Results}, several numerical results are provided and a performance comparison w.r.t. the baseline solution is shown. The conclusions then follow.
For reading fluency, most of the theorem proofs are reported in the appendices. 

\begin{figure}[t]
\centering
\includegraphics[width=\figureWidthScheme]{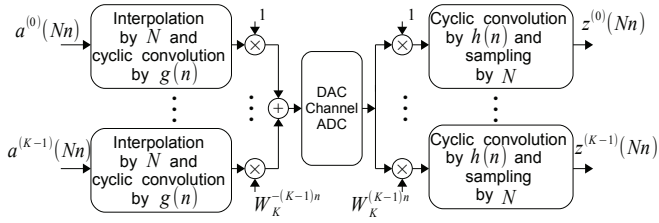}
   \caption{Schematic representation of the CB-FMT transceiver.}
   \label{fig:CBFMT-baseline}
\end{figure}


\section{Cyclic Block FMT Modulation}
\label{Sec:CBFMT}
Cyclic Block Filtered Multitone Modulation is a FBM scheme whose schematic representation is shown in Fig. \ref{fig:CBFMT-baseline} \cite{CBFMT_EURASIP}. The useful signals, constants and operators are listed in Table I.  

CB-FMT splits the transmission of a broad-band information signal into $K$ parallel narrow-band uniformly spaced signals with confined spectrum. In detail, the information data signal, $a(\ell T), \, \ell \in \Z$, (with alphabet belonging to the QAM data set) is serial-to-parallel converted to obtain the low data-data-rate signals $a^{(k)}(\ell NT), \, k \in \{0,\dots,K-1\}$, where $T$ is the sampling period. In the following, a normalized sampling period is assumed, i.e., $T=1$. Differently from conventional FMT, where transmission takes place continuously, in CB-FMT the low data-rate streams are grouped in blocks of $L$ data symbols. The system processes $KL$ data symbols in $MT$ seconds, where $M=LN$. Each symbol in the sub-channel block is interpolated by a factor $N$ and, then, cyclically convolved with a baseband prototype pulse $g(n)$. A multiplication with a complex exponential is performed to translate in frequency the sub-channel signals. Finally, all the sub-channel signals are summed together to yield the transmitted signal $x(n)$:
\begin{align}
	x(n) &= \sum_{k=0}^{K-1} \left[ a^{(k)} \otimes g \right](n) \notag \\	
	&= \sum_{k=0}^{K-1}\sum_{\ell=0}^{L-1}a^{(k)}(\ell N)g((n-\ell N)_{\scriptscriptstyle M}) W_K^{-nk}, \label{eq:CBFMT_TX}\\ 
	& n\in\{0,\cdots,M-1\},\notag
\end{align}
where $\otimes$ denotes the cyclic convolution operator\footnote{The circular convolution between two periodic signals (with $M$ period) $x(n)$ and $y(n)$ is defined as $x \otimes y (n) = \sum_{m=0}^{M-1} x(m) y((n-m)_M), n \in \left[0,\dots,M-1 \right]$.}, $g((n)_{\scriptscriptstyle M})$ is the prototype pulse periodic repetition, i.e., $g((n+aM)_{\scriptscriptstyle M})=g((n)_{\scriptscriptstyle M}), a \in \Z$ and $W_K^{-nk}=e^{j 2 \pi n k/K}$. The prototype pulse is a causal finite impulse response (FIR) filter with $M$ coefficients. If the coefficients number is less than $M$, the pulse can be extended to $M$ with zero-padding, without loss of generality. The modulator processes a block of $KL$ symbols to yield the signal in \eqref{eq:CBFMT_TX}. When more blocks of symbols are transmitted, the overall transmitted signal is obtained with the concatenation of the transmitted blocks of samples $x(n)$.

\renewcommand{\arraystretch}{1}
\begin{table}[t]
  \caption{Useful signals, constants and operators.}
  \label{tab:notation}
  \centering
  \begin{tabular}{|lc|}
  \hline
  \multicolumn{2}{ |c| }{Signals, constants and operators} \\ \hline \hline
  $K$ & number of sub-channels\\
  $N$ & sampling-interpolation factor\\
  $M$ & prototype pulse length\\
  $L=M/N$ & block size\\
  $Q=M/K$ & number of FD coefficients for each sub-channel\\
  $g(n)$ &  synthesis bank prototype filter\\
  $h(n)$ &  analysis bank prototype filter\\
  $g^{(k)}(n)$ &  $g(n) W_K^{-nk}$\\
  $h^{(k)}(n)$ &  $h(n) W_K^{-nk}$\\
  $W_K^{-n}$ & $\exp \left( i 2 \pi n /K \right)$\\
  $G(i)$ & DFT of the synthesis filter bank prototype filter\\
  $H(i)$ &  DFT of the analysis filter bank prototype filter\\
  $(A)_B$ & modulo operation $\left[ A-\text{floor}\left(A/B \right) B \right]$\\
  $\left\{ \mathbf{v} \right\}_i$ & $i$-th element of the vector $\mathbf{v}$\\
  $\tau^a \left\{ \mathbf{v} \right\}$ &  vector cyclic shift $\left( \left\{ \tau^a \left\{ \mathbf{v} \right\} \right\}_i = \left\{ \mathbf{v} \right\}_{(i+a)_N} \right)$\\
  $\left\{ \mathbf{A} \right\}_{i,j}$ & element at the $i$-th row and $j$-th column of the matrix $\mathbf{A}$\\
  $\left\{ \mathbf{A} \right\}_{*,j}$ & $j$-th column of the matrix $\mathbf{A}$\\
  \hline  
  \end{tabular}
\end{table}

The receiver comprises a cyclic analysis filter bank. Thus, the received signal $y(n)$ is multiplied with a bank of $K$ complex exponential functions. The multiplier outputs are cyclically filtered with the prototype analysis pulse $h(n)$. Then, the filter outputs are sampled by a factor $N$ to yield the receiver FB output signals. The $m$-th sample of the $i$-th sub-channel signal can be written as
\begin{align}
	z^{(i)}(mN) &= \sum_{\ell=0}^{M-1} y(\ell)W_K^{\ell i} h((mN - \ell)_{\scriptscriptstyle M}), \label{eq:CBFMT_RX} \\
	i & \in \{0,\dots,K-1\}, \quad m  \in \{0,\dots,L-1\}, \notag
\end{align}  
where $h((n)_{\scriptscriptstyle M})$ is the prototype pulse periodic repetition. To detect the transmitted data symbol, the sub-channel signals in \eqref{eq:CBFMT_RX} are processed with a decision element, e.g., a 1-tap (single coefficient) equalizer.

If the transmission is over a dispersive channel, the signal in \eqref{eq:CBFMT_TX} can be extended with a cyclic prefix (CP). If the CP length is greater than the channel response duration, the linear convolution between the transmitted signal and the channel becomes, locally, a cyclic convolution. This can be exploited to perform a simple frequency domain channel equalization as shown in \cite{CBFMT_EURASIP} and in Section \ref{Sec:Practical}. In general, assuming a CP of $\mu$ samples, the transmission data rate is equal to
\begin{equation}
R = \frac{KL}{(M+\mu)T} \quad \mbox{symbols/s}. \label{eq:CBFMT_RATE}
\end{equation}



\subsection{Perfect Reconstruction Conditions}
\label{Sec:CBFMT:PR}
In this section, we will discuss the design of a perfect reconstruction (PR) CB-FMT system, i.e., the perfect reconstruction of the cyclic FB. 
Herein, the communication medium is assumed ideal so that $y(n)=x(n)$. A real transmission medium is considered in Sec. \ref{Sec:Practical}. If the PR conditions are fulfilled, neither inter-channel interference (ICI) nor inter-symbol interference (ISI) will be exhibited at the analysis filter bank output. 

By replacing \eqref{eq:CBFMT_TX} in \eqref{eq:CBFMT_RX}, after some algebraic manipulation, we obtain
\begin{equation}
 \tilde{z}^{(i)}(m N) = \sum_{k=0}^{K-1} \sum_{\ell=0}^{L-1} \tilde{a}^{(k)}(\ell N) [g^{(k)} \otimes h^{(i)}](mN-\ell N), \label{eq:CBFMT_IN_OUT}
 \end{equation}
where 
\begin{align}
\tilde{a}^{(k)}(\ell N) =& a^{(k)}(\ell N) W_K^{-\ell N k},\label{eq:atilde}\\
\tilde{z}^{(i)}(m N) =& z^{(i)}(m N) W_K^{m N i},\label{eq:ztilde}\\ 
g^{(k)}(n) =& g\left((n)_{\scriptscriptstyle M}\right) W_K^{-nk}, \label{eq:gk2} \\ 
h^{(i)}(n) =& h\left((n)_{\scriptscriptstyle M}\right) W_K^{-ni}. \label{eq:hh2}
\end{align}
Equations \eqref{eq:atilde}--\eqref{eq:ztilde} represent a rotation of the symbol constellation, while \eqref{eq:gk2}--\eqref{eq:hh2} represent a translation of the prototype pulse in the frequency domain. In \eqref{eq:CBFMT_IN_OUT}, the cyclic convolution between $g^{(k)}(n)$ and $h^{(i)}(n)$ corresponds to the cyclic cross-correlation between $g^{(k)}(n)$ and $\left( h^{(i)}(n) \right)^*$, sampled by a factor $N$. Relation \eqref{eq:CBFMT_IN_OUT} can be rewritten as
\begin{align}
 \tilde{z}^{(i)}(m N) =& \sum_{k=0}^{K-1} \sum_{\ell=0}^{L-1} \tilde{a}^{(k)}(\ell N) r^{(k,i)}(mN - \ell N), \label{eq:CBFMT_IN_OUT_CORRELATION}\\
 r^{(k,i)}(mN) = & g^{(k)} \otimes h^{(i)} (mN), \label{eq:rgh_TD}
 \end{align}
where \eqref{eq:rgh_TD} represents the cyclic cross-convolution between the synthesis and the analysis sub-channel pulses. In the following, it will be referred to as cyclic cross-convolution function (CCF). When $k=i$ it will be referred to as cyclic auto-convolution function (ACF). The CCF is a periodic function with period $M$. 

CB-FMT has PR if and only if the Generalized Nyquist Criterion (GNC) \cite{generalized_nyquist} is satisfied:
\begin{enumerate}
\item For $k=i$, the ACF is a Kronecker delta, i.e., $r^{(i,i)}(mN)=\delta_m$. Thus, for each sub-channel there is no interference between different symbols in the same block (no ISI condition).
\item For $k \neq i$, the CCF is always null. Thus, there is no interference between different sub-channels (no ICI condition).
\end{enumerate}

In the following, it will be shown that the pulse design is simplified if we operate in the frequency domain. Therefore, it is important to state the PR conditions in the frequency domain. To start, the DFT of the CCF in \eqref{eq:rgh_TD} is computed to obtain 
\begin{equation}
R^{(k,i)}(p) = \frac{1}{N} \sum_{s=0}^{N-1} G(p+sL+kQ) H(p+sL+iQ), \label{eq:rg_FD_4}
\end{equation}
where $Q=M/K$ is a positive integer number while $G(p)$ and $H(p)$ are the $M$-point DFT of the transmitter and receiver prototype pulses.

For the CB-FMT system, the frequency domain translation of the GNC can be written as follows.
\begin{thm}[\bfseries CB-FMT frequency domain GNC] \label{thm:PRFD}
CB-FMT has PR if and only if the following two frequency domain conditions are satisfied:
\begin{enumerate}
\item For $k=i$, the DFT of the ACF is a constant; analytically, 
\begin{align}
 R^{(k,k)}(p) = & \notag \\
 \frac{1}{N} \sum_{s=0}^{N-1} G&(p+sL+kQ) H(p+sL+kQ)=1, \label{eq:PR_FD_1}\\
      \forall p  \in \{0,& \dots,L-1\},  \quad \forall k \in \{0,\dots,K-1\}. \notag
\end{align}
\item For $k \neq i$, the DFT of the CCF is always null; analytically, we have
\begin{align}
 R^{(k,i)}(p) = & \notag\\
 \frac{1}{N} \sum_{s=0}^{N-1} G&(p+sL+kQ) H(p+sL+iQ) =0, \label{eq:PR_FD_2}\\
      \forall p  \in \{0,& \dots,L-1\}, \quad \forall k, i \in \{0,\dots,K-1\}. \notag
\end{align}
\end{enumerate}
\end{thm}
\begin{proof}
The proof is immediate: the frequency domain PR conditions are obtained with a transform of the time domain PR conditions into the frequency domain. 
\end{proof}

\section{Orthogonality Conditions} 
\label{Sec:Orthogonality}
In the presence of Gaussian additive background noise, the SNR is maximized when the transmitter and receiver filters are matched \cite{ref:Proakis}, i.e., $h(n)=g^*_{-}(n) = g^*(-n)$. With matched analysis pulses, if the PR conditions are fulfilled the Cyclic FB will be orthogonal. Orthogonality conditions are a particular case of the PR conditions introduced in Sec. \ref{Sec:CBFMT:PR}. For clarity, the orthogonality conditions are reported in the following as a corollary of Thm. \ref{thm:PRFD}. 

\begin{cor}[Frequency domain orthogonal conditions] \label{cor:FDort}
CB-FMT is orthogonal if and only if the following conditions are fulfilled:
\begin{enumerate}
\item For $k=i$, the ACF has flat DFT spectrum, i.e., 
\begin{align}
 R^{(k,k)}(p) = \frac{1}{N} &\sum_{s=0}^{N-1} |G(p+sL+kQ)|^2=1, \label{eq:ort_FD_1}\\
      \forall p  \in \{0,& \dots,L-1\},  \quad \forall k \in \{0,\dots,K-1\}. \notag
\end{align}
\item For $k \neq i$, the DFT of the CCF is always null, i.e., 
\begin{align}
R^{(k,i)}(p) = & \notag\\
 \frac{1}{N} \sum_{s=0}^{N-1}& G(p+sL+kQ) G^*(p+sL+iQ) =0, \label{eq:ort_FD_2}\\
      \forall p  \in \{0,& \dots,L-1\}, \quad \forall k,i \in \{0,\dots,K-1\}. \notag
\end{align}
\end{enumerate}
\end{cor}
The periodic property of the DFT allows us to rewrite \eqref{eq:ort_FD_2} as
\begin{align}
 \frac{1}{N} \sum_{s=0}^{N-1} G(p+sL)& G^*(p+sL+kQ) =0, \label{eq:ort_FD_3}\\
      \forall p  \in \{0, \dots,L-1\},& \quad \forall k \in \{1,\dots,K-1\}. \notag
\end{align}
Eq. \eqref{eq:ort_FD_3} is equivalent to \eqref{eq:ort_FD_2} but it depends only on the two variables ($p$ and $k$). The use of \eqref{eq:ort_FD_3} instead \eqref{eq:ort_FD_2} enables us to obtain the orthogonality conditions in matrix form.

\subsection{Orthogonality Conditions in Matrix Form}
\label{Sec:Orthogonality:MatrixForm}
In pulse design, a common practice is to write the orthogonality conditions in matrix form. This form allows to rewrite the conditions \eqref{eq:ort_FD_1} and \eqref{eq:ort_FD_3} as a set of non-linear systems. These systems will be exploited to design the pulses and to show that the number of orthogonal pulses is infinite.

To proceed, \eqref{eq:ort_FD_3} can be viewed as an Hermitian inner product between two vectors $\mathbf{a}$ and $\mathbf{b}$,
\begin{equation}
\mathbf{a} \cdot \mathbf{b} = 0, \label{eq:inner_product}
\end{equation}
where the two vectors are defined as
\begin{align}
\mathbf{a} &= \left[ G(p), G(p+L),\dots \right]^T, \label{eq:va}\\
\mathbf{b} &= \left[ G(p+kQ), G(p+L+kQ),\dots \right]^T. \label{eq:vb}
\end{align}
The vector $\mathbf{a}$ can be easily rewritten as
\begin{align}
\left\{ \mathbf{a} \right\}_i = \left\{ \mathbf{v}_p \right\}_i &\doteq G(p+iL), \label{eq:vp} \\
p \in \{0,\dots,L-1\}, & \quad i \in \{0,\dots,N-1\}. \notag 
\end{align}
In \eqref{eq:vp}, the operator $\left\{ \mathbf{v}_p \right\}_i$ (see Tab. \ref{tab:notation}) yields the $i-th$ element of the vector $\mathbf{v}_p$. The definition of $\mathbf{v}_p$ introduces a partition of the $M$-point DFT coefficients of the prototype pulse into $L$ vectors of size $N \times 1$.

The vector $\mathbf{b}$ in \eqref{eq:vb} can be written as
\begin{align}
\mathbf{b} &= \left[ G\left(c_{(p,k)}+d_{(p,k)}L\right), G\left(c_{(p,k)}+d_{(p,k)}L+L\right),\dots \right]^T. \label{eq:vb2}
\end{align}
where 
\begin{align}
c_{(p,k)} &= (p+kQ)_L, \label{eq:cpk}\\
d_{(p,k)} &= \frac{p+kQ-c_{(p,k)}}{L}=\frac{p+kQ-(p+kQ)_L}{L}. \label{eq:dpk}
\end{align}

Using the vector cyclic shift operator (see Tab. \ref{tab:notation}), \eqref{eq:vb2} can be written as $\mathbf{b} = \tau^{d_{(p,k)}} \left\{ \mathbf{v}_{c_{(p,k)}} \right\}$. For a given $p$, \eqref{eq:ort_FD_3} is the result of $K$ inner products, each between the vector $\mathbf{v}_p$ and its shifted version $\tau^{d_{(p,k)}} \left\{ \mathbf{v}_{c_{(p,k)}} \right\}$ for a certain $k$. These $K$ vectors,
\begin{equation}
\left\{ \mathbf{v}_p, \, \tau^{d_{(p,1)}} \left\{ \mathbf{v}_{c_{(p,1)}} \right\}, \, \dots, \, \tau^{d_{(p,K-1)}} \left\{ \mathbf{v}_{c_{(p,K-1)}} \right\} \right\},
\end{equation}
can be gathered in a matrix of size $K \times N$ defined as
\begin{align}
\mathbf{H}_{\text{ort},p} &= \frac{1}{\sqrt{N}} \hat{\mathbf{H}}_p, \label{eq:Hort}\\
\left\{ \hat{\mathbf{H}}_p \right\}_{*,j}, &= \tau^{d_{(p,j)}} \left\{ \mathbf{v}_{c_{(p,j)}} \right\}, \label{eq:Hort2}\\
p \in \{0, \dots, L-1 \}, & \quad j \in \{0,\dots, K-1\}. \notag
\end{align}

Now, the orthogonality conditions can be stated in matrix form.
\begin{thm} \label{thm:Matrix}
The orthogonality conditions in \eqref{eq:ort_FD_1} and \eqref{eq:ort_FD_2} are satisfied if and only if the matrices defined in \eqref{eq:Hort} have orthonormal columns for any $p \in \left\{ 0, \dots, N_s-1 \right\}$, where $N_s=\gcd(Q,L)$. 
\end{thm}
\begin{proof}
The proof in reported in Appx. \ref{Appx:MatrixForm}.A.
\end{proof}

The results of Thm. \ref{thm:Matrix} say that in general the $M$ unknowns (the filter coefficients) are partitioned in $L$ vectors of size $N \times 1$. Then, these vectors can be grouped in $N_s$ sets, one for every matrix $\mathbf{H}_{\text{ort},p}$, $p\in \left\{0,..., N_s-1\right\}$. These matrices are composed by $L/N_s$ distinct $ \mathbf{v}_p$ vectors and their circular shifted versions. Furthermore, the matrices are disjoint, i.e., they contain different sets of unknowns. Therefore, to obtain orthogonality, the $N_s$ sets of relations, each with $M/N_s$ unknowns, can be solved independently.

An interesting case is the critically sampled CB-FMT system, for which $K=N$. This system offers the maximum transmission rate. In this case, the matrices are circulant as stated and proved in the following corollary.

\begin{cor}[Critically sampled case] \label{cor:critically}
When the system is critically sampled, i.e., $K=N$ and $Q=L$, the $N_s= L$ matrices $\mathbf{H}_{\text{ort},p}$ are circulant matrices. In this case, the matrices are orthogonal if and only if the vectors \eqref{eq:vp} have unit modulus $N$-point DFT, i.e., $\left| \left\{\mathbf{F}_N \mathbf{v}_p \right\}_i \right|=1$, where $\mathbf{F}_N$ is the $N \times N$ DFT matrix.
\end{cor}
\begin{proof}
The proof in reported in Appx. \ref{Appx:MatrixForm}.B.
\end{proof}

In conclusion, when Thm. \ref{thm:Matrix} is satisfied, the prototype pulse is orthogonal. The set of orthogonal pulses has infinite cardinality, as stated in the following theorem.
\begin{thm} \label{thm:inf}
Given a system with parameters $(K, N, M)$, there exists an infinite number of prototype pulses that satisfy the Thm. \ref{thm:Matrix} orthogonality conditions. 
\end{thm}
\begin{proof}
The proof in reported in Appx. \ref{Appx:inf}.
\end{proof}

Two objective functions are introduced in Sec. \ref{Sec:Design} to constrain the search and obtain optimal pulses that satisfy \ref{thm:Matrix}.

\subsection{Orthogonality under Parameters Variation}
\label{Sec:Orthogonality:Params_variation}
In this section, we discuss whether an orthogonal pulse in a CB-FMT system with a given set of parameters $(K, N, M)$ can be "reused" as a mother pulse once the parameters are varied.  When the prototype pulse is frequency confined, i.e., $G(i)=0 \text{ for } i \in \left[ Q, \dots, M-1 \right]$, its FD coefficients can be used to obtain an orthogonal pulse in a system where the parameters $(K, N, M)$ are all increased by a factor $\alpha_1$ (which corresponds to an increase of the pulse length) or when the number of sub-channels $K$ is increased by a factor $\alpha_2$ while $M$ is kept constant (which corresponds to maintain the pulse length constant). These results are stated in the following two theorems and a graphical representation is depicted in Fig. \ref{fig:params_extension}.    

\begin{figure}[t]
\centering
\includegraphics[width=3.4in]{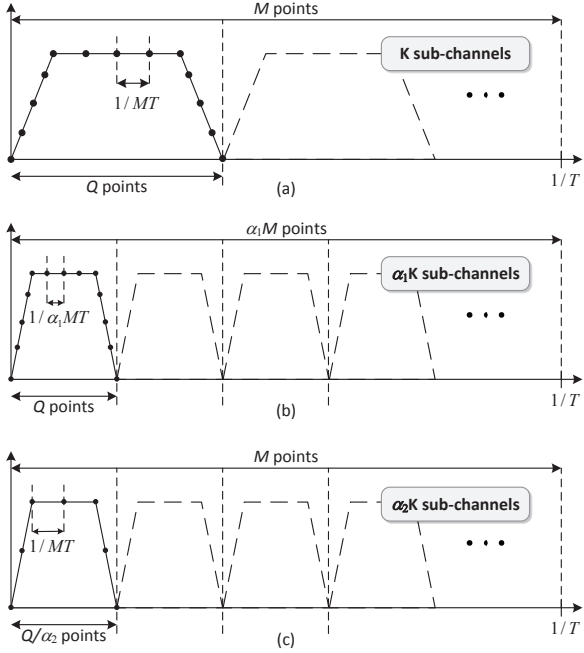}
   \caption{Graphical representation of the orthogonal pulse spectrum under a system parameters variation. In (a), the mother pulse designed for the set of parameters $(K,N,M)$. In (b), Thm. \ref{thm:params_1} allows to obtain an orthogonal pulse for the set of parameters $(\alpha_1 K, \alpha_1 N, \alpha_1 M)$. In (c), Thm. \ref{thm:params_2} allows to obtain an orthogonal pulse for the set of parameters $(\alpha_2 K, \alpha_2 N, M)$.}
   \label{fig:params_extension}
\end{figure}

\begin{thm}[\bfseries Filter length variation] \label{thm:params_1}
Given an orthogonal prototype pulse with FD coefficients $G(i)$ that satisfies Cor. \ref{cor:FDort} and designed for a set of parameters $(K, N, M)$, the prototype pulse defined as
\begin{align} \label{eq:Galpha_1}
G_{\alpha_1}(i) &= \begin{cases} 
\sqrt{\alpha_1} G(i) & \text{for } i \in \left[ 0,\dots,Q-1 \right]\\
0 & \text{otherwise}
\end{cases}\\
 i & \in \left[0,\dots, \alpha_1 M -1 \right] \notag
\end{align}
is orthogonal for the following set of parameters: 
\begin{equation} \label{eq:params_alpha_1}
(\alpha_1 K, \alpha_1 N, \alpha_1 M),
\end{equation}
where $\alpha_1 \in \R$ is a constant s.t. all the parameters in \eqref{eq:params_alpha_1} are integer numbers.
\end{thm}
\begin{proof}
The proof is reported in Appx. \ref{Appx:params_1}.
\end{proof}

\begin{thm}[\bfseries Constant filter length] \label{thm:params_2}
Given an orthogonal prototype pulse that satisfies Cor. \ref{cor:FDort} and designed for a set of parameters $(K, N, M)$, the prototype pulse defined as
\begin{align} \label{eq:Galpha_2}
G_{\alpha_2}(i) &= \begin{cases} 
 \sqrt{\alpha_2} G(\alpha_2 i) & \text{for } i \in \left[ 0,\dots,Q/\alpha_2-1 \right]\\
0 & \text{otherwise}
\end{cases}\\
 i & \in \left[0,\dots, M-1 \right] \notag
\end{align}
is orthogonal for the following set of parameters: 
\begin{equation} \label{eq:params_alpha_2}
(\alpha_2 K, \alpha_2 N, M),
\end{equation}
where $\alpha_2 \in \Z^+$.
\end{thm}
\begin{proof}
The proof is reported in Appx. \ref{Appx:params_2}.
\end{proof}



\section{Transmission Medium Effects on Orthogonality}
\label{Sec:Practical}
The orthogonality discussed in previous sections, assures that the cyclic FB is orthogonal when the transmission medium is ideal, i.e., $y(n) = x(n)$. In this section, the effect of a transmission medium that is not ideal is considered. In general, we assume it to be modeled with a linear and time variant filter with impulse response $g_\text{eq}(n, m)$. Two specific situations are envisioned: a) the equivalent filter is time invariant and it models the effects of the D/A-A/D converter filters and of a dispersive medium; b) the equivalent filter is time variant and it models a medium that exhibits time variant and frequency selective fading as in mobile wireless communication channels. It follows that the received signal can be written as 
\begin{align}
 y(n) &= x*g_\text{eq}(n) \notag\\
 &= \sum_{m=0}^{P-1} g_\text{eq}(n, n-m) x(n-m), \label{eq:geq2}\\
 g_\text{eq}(n, m) &= \sum_{s=0}^{P-1} \alpha_s(n) \delta(m-s) \label{eq:getv}
 \end{align} 
where $*$ is the linear convolution operator. In \eqref{eq:getv}, $P$ and $\alpha_s(n)$ are the impulse response length (in samples) and the impulse response time variant coefficients, respectively. In general, the equivalent filter may cause an orthogonality loss in the filter bank. To simplify the equalization task, a cyclic prefix (CP) can be added to each block of coefficients in \eqref{eq:CBFMT_TX}, similarly to CP-OFDM. If the CP length (in samples) is greater than the equivalent filter length $P$ (with $P < \mu < M$), the convolution in \eqref{eq:geq2} becomes cyclic w.r.t. the variable $m$. Thus,
\begin{align}
 y(n) = x \otimes g_\text{eq}(n). \label{eq:geq3}
 \end{align}
 
 
\subsection{Linear Time Invariant Medium}
When the equivalent filter is time invariant, orthogonality is maintained if the following theorem is satisfied.
\begin{thm}[\bfseries Orthogonality with equiv. filter] \label{thm:geq}
An orthogonal CB-FMT system keeps its orthogonality in the presence of an equivalent filter between the transmitter and the receiver, if and only if the filter $g_\text{eq}(n)$ is such that $g_\text{eq} \otimes g_\text{eq,-}^*(mN)=\delta_m$, the receiver pulse in \eqref{eq:hh2} is equal to $h^{(i)}(n)= \left( g_\text{eq,-} \otimes g^{(i)}_-(n) \right)^*$ and the CP length is greater than the equivalent filter length $P$.
\end{thm}
\begin{proof}
The proof is reported in Appx. \ref{Appx:geq}
\end{proof}

Generally, the equivalent filter does not satisfy Thm. \ref{thm:geq}. In this case, to restore the orthogonality an equalizer is required. The cyclic convolution suggests a frequency domain equalization \cite{CBFMT_ICC}. For each received block of samples, the CP is disregarded and the $M$-point DFT of the signal \eqref{eq:geq3} is computed to obtain
\begin{equation}
Y(q) = G_{\text{eq}}(q) X(q),
\end{equation}
where $ G_{\text{eq}}(q)$ is the $M$-point DFT of the time invariant equivalent filter and $X(q)$ is the $M$-point DFT of the transmitted signal in \eqref{eq:CBFMT_TX}. In this case,  a 1-tap equalizer is sufficient to restore orthogonality. The equalizer output signal reads
\begin{align}
Y_\text{eq}(q) &= C_\text{tinv}(q) Y(q) = X(q),\\
C_\text{tinv}(q) &= 1/G_{\text{eq}}(q), \label{eq:zf_equalizer}
\end{align}
where $C_\text{tinv}(q)$ is the time-invariant equalizer coefficient. The equalizer in \eqref{eq:zf_equalizer} is known as Zero-Forcing (ZF) equalizer. This equalizer suffers from the problem of noise enhancement for very small $G_{\text{eq}}(q)$ coefficients. To solve this issue, the MMSE equalizer can be adopted \cite{CBFMT_EURASIP}:
\begin{equation}
C_\text{tinv, MMSE}(q) = \frac{G^*_{\text{eq}}(q)}{|G_{\text{eq}}(q)|^2 + \sigma^2_n/|G(q)|^2},
\end{equation}
where $\sigma^2_n$ is the noise variance.

After the equalization, filtering in the FD with the prototype analysis pulse $H(q)=G^*(q)$ is performed and finally an M-point IDFT is applied to obtain
\begin{align}
z^{(i)}(mN) = a^{(i)}(mN) + \eta^{(i)}(mN),
\end{align}
where $\eta^{(i)}(mN)$ are the background noise samples in the $i$-th sub-channel.

\subsection{Linear Time Variant Medium}
\label{Sec:Practical:TV}
When the equivalent filter is time variant, e.g., as in a mobile radio channel, the orthogonality is lost. If we still assume the CP to be longer than the channel impulse response duration, the $M$-point DFT of \eqref{eq:geq3} yields 
\begin{equation}
Y(q) = \sum_{p=0}^{M-1} X(p) H_2(p,q-p),
\end{equation}
where $H_2(p,q)$ is the two-dimensional $M$-point DFT of the equivalent filter coefficients $\alpha_s(n)$ in \eqref{eq:getv} and it is defined as
\begin{equation}
H_2(p,q) = \sum_{s=0}^{M-1} \sum_{n=0}^{M-1} \alpha_s(n) W_M^{sp+nq}.
\end{equation}
Details can be found in \cite{CBFMT_EURASIP}. Assuming to deploy a simple 1-tap equalizer also in this situation, the equalizer output signal can be written as 
\begin{align}
Y_\text{eq}(q) &= C_\text{tvar}(q) Y(q) \notag\\
&= X(q)+ C_\text{tvar}(q) \sum_{\substack{p=0\\p \neq q}}^{M-1} X(p) H_2(p,q-p)\\
& =  X(q) + I(q) \label{eq:interf}
\end{align}
where the MMSE equalizer coefficients $C_\text{tvar}(q)$ are computed as described in \cite{ref:CBFMT_EUSIPCO}.  This equalizer does not restore the orthogonality and some interference may remain, as shown in \eqref{eq:interf} with the additive term $I(q)$. In fact, if, as a final step, we perform matched filtering with the analysis pulse $G^*(q)$ and we compute the M-point IDFT, the output sub-channel signal will read
\begin{align}
z^{(i)}(mN) =& a^{(i)}(mN) \notag \\
&+ \sum_{\substack{\ell=0 \\ \ell \neq m}}^{L-1} a^{(i)}(\ell N) r_{\text{interf}}^{(i,i)}(mN-\ell N)  \notag \\
& + \sum_{\substack{k=0 \\ k \neq i}}^{K-1} \sum_{\ell=0}^{L-1} a^{(k)}(\ell N) r_{\text{interf}}^{(k,i)}(mN-\ell N)  \notag \\
& + \eta^{(i)}(mN), \label{eq:ztv}
\end{align}
where $r_{\text{interf}}^{(k,i)}(mN)$ are the interference coefficients and the convolutions in \eqref{eq:ztv} are cyclic with period $M$.

The level of such a residual interference depends on the specific prototype pulse used. In the following, we will design orthogonal CB-FMT pulses that allow to maximize the system capacity in the presence of a time-variant channel. 

Others equalization schemes can be adopted. As an example, in \cite{CBFMT_EURASIP} a multi-channel equalizer is considered and the coefficients are obtained by jointly considering all sub-channels, i.e., the inter-channel interference. This is more complex but it can improve performance. 



\section{Pulse Design}
\label{Sec:Design}
In this section, we constructively exploit Thm. \ref{thm:Matrix} to design an orthogonal CB-FMT system. The design process is partitioned in three parts. First, a parameterization with angles is introduced to reduce the number of unknowns. Then, the non-linear system is introduced. The solution of this system allows to satisfy the orthogonality conditions. Finally, since the orthogonal pulses are infinite, two objective functions are introduced to determine an optimal pulse w.r.t. the selected metric.

\subsection{Parameterization with Angles}
\label{Sec:Design:Angles}
To achieve orthogonality, the column vectors of $\hat{\mathbf{H}}_{\text{ort},p}$ (with elements made by the pulse coefficients) must have unit norm, i.e., $||\mathbf{v}_p||^2/N=1$. This suggests to express the vector components in terms of non-linear combination of trigonometric functions. 

Firstly, we focus on real valued solutions, i.e., the pulse FD coefficients are real so that $\mathbf{v}_p \in \R^N$. For $N=2$, the square norm is simply given by $x^2+y^2=1$. A solution is given by $x=\cos \alpha, y=\sin \alpha$. This representation exploits the polar coordinates with unit radius. For $N=3$, the square norm $x^2+y^2+z^2=1$ represents a sphere in $\R^3$ with unit radius. The exploitation of the spherical coordinates allows to write the solution as $x=\cos \alpha, y = \sin \alpha \cos \beta, z=\sin \alpha \sin \beta$. In general, for $\R^N, N>3$, the hyper-spherical coordinates \cite{sommervilleGeometry} can be used. 

Secondly, the pulse FD coefficients can be complex so that the vectors $\mathbf{v}_p$ are complex valued. Thus, a phase factor can be added to every vector components and a representation with angles can be used
\begin{align}
\left\{ \mathbf{v}_p \right\}_0 &= \sqrt{N} \cos \left( \left\{ \boldsymbol{\theta}_p \right\}_0 \right)e^{j \left\{ \boldsymbol{\Phi}_p \right\}_0 }, \label{eq:angles3}\\
\left\{ \mathbf{v}_p \right\}_1 &= \sqrt{N} \sin \left( \left\{ \boldsymbol{\theta}_p \right\}_0 \right) \cos \left( \left\{ \boldsymbol{\theta}_p \right\}_1 \right)e^{j \left\{ \boldsymbol{\Phi}_p \right\}_1 },\\
\vdots \quad \; &= \qquad \qquad \; \vdots \notag\\
\underset{i \in \{2,\dots, N-2\}}{\left\{ \mathbf{v}_p \right\}_i} &= \sqrt{N} \left[ \prod_{s=0}^{i-1} \sin \left( \left\{ \boldsymbol{\theta}_p \right\}_s \right) \right] \cos \left( \left\{ \boldsymbol{\theta}_p \right\}_i \right)e^{j \left\{ \boldsymbol{\Phi}_p \right\}_i },  \\
\vdots \quad \; &= \qquad \qquad \; \vdots \notag\\
\left\{ \mathbf{v}_p \right\}_{N-1} &= \sqrt{N} \left[ \prod_{s=0}^{N-2} \sin \left( \left\{ \boldsymbol{\theta}_p \right\}_s \right) \right]e^{j \left\{ \boldsymbol{\Phi}_p \right\}_{N-1} }, \label{eq:angles4}
\end{align}
where the angles $\boldsymbol{\theta}_p$ and $\boldsymbol{\Phi}_p$ identify the amplitude and phase of the vector $\mathbf{v}_p \in \C^N$ components. The total number of angles (unknowns) is $L(2N-1)$ since we have $L$ vectors with $N$ components each.

\subsection{Non-Linear System}
\label{Sec:Design:System}
The angles representation assures that the condition \eqref{eq:ort_FD_1} is satisfied for any value of the angles so that no ISI is present. To complete the orthogonal pulse design, we have to apply Thm. \ref{thm:Matrix} which can be explicitly written to obtain the following non-linear system of equations
\begin{equation}
\begin{cases} 
\left\{ \mathbf{H}_{\text{ort},p} \right\}_{*,0}^H \cdot \left\{ \mathbf{H}_{\text{ort},p} \right\}_{*,i}=0 & i \in \{1,\dots, K-1\}\\
\left\{ \mathbf{H}_{\text{ort},p} \right\}_{*,1}^H \cdot \left\{ \mathbf{H}_{\text{ort},p} \right\}_{*,i}=0 & i \in \{2,\dots, K-1\}\\
\qquad \qquad \quad \vdots & \qquad \qquad \vdots\\
\left\{ \mathbf{H}_{\text{ort},p} \right\}_{*,K-3}^H \cdot \left\{ \mathbf{H}_{\text{ort},p} \right\}_{*,i}=0 & i \in \{K-2, K-1\}\\
\left\{ \mathbf{H}_{\text{ort},p} \right\}_{*,K-2}^H \cdot \left\{ \mathbf{H}_{\text{ort},p} \right\}_{*,K-1}=0 & 
\end{cases}. \label{eq:orthsystem}
\end{equation}
From equation \eqref{eq:orthsystem}, we should note that the system comprises $\sum_{k=0}^{K-1} k = K(K-1)/2$ equations of $L (2N-1)/N_s$ unknown angles for complex pulse solutions and $L (N-1)/N_s$ for real pulse solutions. The problem comprises $N_s$ independent systems, thus the total number of equations is equal to $N_s K(K-1)/2$ that are split in $N_s$ distinct sub-systems.

In general, the components of a vector $\mathbf{v}_p$ are represented with a set of $2N-1$ angles. In fact, from \eqref{eq:angles3}-\eqref{eq:angles4}, the components are represented by $N-1$ angles for the amplitude ($\boldsymbol{\theta}_p$) and $N$ angles for the phases ($\boldsymbol{\Phi}_p$). The angles number can be reduced if a band limited constraint to the pulse is set. We denote with $Q_2$ the number of non-zero DFT coefficients of the pulse, i.e., $G(i)=0$ for $Q_2 < i \leq M$. In this case, some components of the vectors $\mathbf{v}_p$ are equal to zero and the angles representation -- always from \eqref{eq:angles3}-\eqref{eq:angles4} -- is used only for the non-zero components. If we assume that $\mathbf{v}_p$ has only $N_2 = M/Q_2 < N$ non-zero components than we obtain a set of $2N_2-1$ angles, $N_2-1$ for the amplitudes and $N_2$ for the phases.

When $Q_2 = Q$, the prototype pulse is confined into the sub-channel and it does not overlap the adjacent sub-channels. In this case, the non-linear system \eqref{eq:orthsystem} is automatically satisfied. Thus, there are an infinite set of angles and phases that satisfy the equations \eqref{eq:angles3}--\eqref{eq:angles4}. Two cases can be distinguished:
\begin{enumerate}
\item \textbf{Over-sampled case ($K<N$)}. In this case, the vector $\mathbf{v}_p$ has at most $\lceil Q/L \rceil$ non-zero components. Thus, the angles sets $\boldsymbol{\theta}_p$ and $\boldsymbol{\Phi}_p$ have at most $2 \lceil Q/L \rceil-1$ angles.
\item \textbf{Critically-sampled case ($K=N$)}. In this particular case $Q=L$ and the vectors $\mathbf{v}_p$ have only one non-zero coefficient. Thus, the only valid solution is written as
\begin{equation}
 \begin{cases}
 G(p) = e^{j \left\{ \boldsymbol{\Phi}_p \right\}_0 } & \text{for } p \in \{0,\dots,L-1\} \\
 \quad 0 & \text{otherwise}
 \end{cases}. \label{eq:critically_rect}
\end{equation}
Condition \eqref{eq:critically_rect} shows that the only possible solution corresponds to a pulse whose frequency response is the rectangular window.
\end{enumerate}



\subsection{Practical Design}
\label{Sec:Design:Practical}
As shown in Thm. \ref{thm:inf}, there is an infinite number of solutions to the orthogonal pulse design problem. To complete the design procedure, we consider it under the goal of maximizing a certain objective function. Two objective functions and metrics have been identified:
\begin{enumerate}
\item the maximum in-band-to-out-of-band pulse energy, referred to as IBOB energy ratio;
\item the maximum achievable rate.
\end{enumerate}

\subsubsection{In-band-to-out-of-band Energy Metric} \label{Sec:Design:Practical:Ratio}
To compute the IBOB energy ratio, we consider the discrete-time Fourier transform (DTFT) of the prototype pulse so that the objective function is defined as
\begin{align}
f_1(\boldsymbol{\theta}, \boldsymbol{\Phi}) &= \frac{\int_0^B \left| S(f,\boldsymbol{\theta},\boldsymbol{\Phi})\right|^2 \mbox{d}f}{\int_{-\infty}^{+\infty} \left| S(f,\boldsymbol{\theta},\boldsymbol{\Phi})\right|^2 \mbox{d}f-\int_0^B \left| S(f,\boldsymbol{\theta},\boldsymbol{\Phi})\right|^2 \mbox{d}f}, \label{eq:f1}\\
\boldsymbol{\theta} &= \left[ \boldsymbol{\theta}_0,\dots, \boldsymbol{\theta}_{L-1} \right], \label{eq:f1_a}\\
\boldsymbol{\Phi} &= \left[ \boldsymbol{\Phi}_0,\dots, \boldsymbol{\Phi}_{L-1} \right], \label{eq:f1_b}
\end{align}
where $S(f,\boldsymbol{\theta},\boldsymbol{\Phi})$ is the DTFT of the prototype pulse, $B=1/KT$ is the sub-channel bandwidth and $\boldsymbol{\theta}$, $\boldsymbol{\Phi}$ are two $(2N-1) \times L$ matrices that contain all the angle parameters of equations \eqref{eq:angles3} to \eqref{eq:angles4}. 

\subsubsection{Maximum Achievable Rate Metric} \label{Sec:Design:Practical:Capacity}
Under Gaussian additive background noise, the maximum achievable rate (capacity) can be computed as follows  
\begin{align}
f_2&(\boldsymbol{\theta}, \boldsymbol{\Phi})= \notag \\
&= \frac{1}{(M+\mu)T} \sum_{k=0}^{K-1}\sum_{\ell=0}^{L-1} \log_2 \left( (1+\mbox{SINR}^{(k)}(\ell, \boldsymbol{\theta}, \boldsymbol{\Phi}) \right),\label{eq:f2}
\end{align}
where $\mbox{SINR}^{(k)}(\ell, \boldsymbol{\theta}, \boldsymbol{\Phi})$ represents the signal-to-noise-plus-interference experienced by the $\ell$-th element of the data block transmitted in the $k$-th sub-channel (see also \eqref{eq:ztv}). The SINR depends on the specific channel realization and on the prototype pulse. 

If we consider a random time variant fading channel, an optimal orthogonal pulse will be found for each specific channel realization, which would require to adapt the filter bank to the channel conditions. In order to maintain a unique FB, we consider to select a unique prototype pulse and in particular the one that maximizes the average capacity. In other words, we proceed as follows: a) for each channel realization we design the optimal capacity wise pulse; b) for all pulses we determine the average capacity; c) we select the pulse that yields the highest average capacity. Clearly the procedure is applicable once a given random channel model is available as discussed in the numerical results section. 

\subsection{Summary of the Design Algorithm} \label{Sec:Design:Practical:Algorithm}      
The design procedure described in the previous sections can be summarized as follows.
\begin{enumerate}
\item Exploiting \eqref{eq:vp}, the $M$ unknowns (the pulse DFT coefficients) are partitioned in $L$ vectors $\mathbf{v}_p$.

\item Exploiting \eqref{eq:angles3}-\eqref{eq:angles4}, the vectors $\mathbf{v}_p$ are expressed in term of angles, $\boldsymbol{\theta}$ and $\boldsymbol{\Phi}$.

\item Exploiting \eqref{eq:orthsystem}, $N_s$ independent non-linear systems of equations are generated.

\item The metrics $f_1(\boldsymbol{\theta}, \boldsymbol{\Phi})$ and $f_2(\boldsymbol{\theta}, \boldsymbol{\Phi})$ are optimized under the orthogonality constraint, namely the non-linear systems of equations.
\end{enumerate}



\section{Numerical Results}
\label{Sec:Results}
\begin{figure}[t]
\centering
\includegraphics[width=\figureWidthFlowchart]{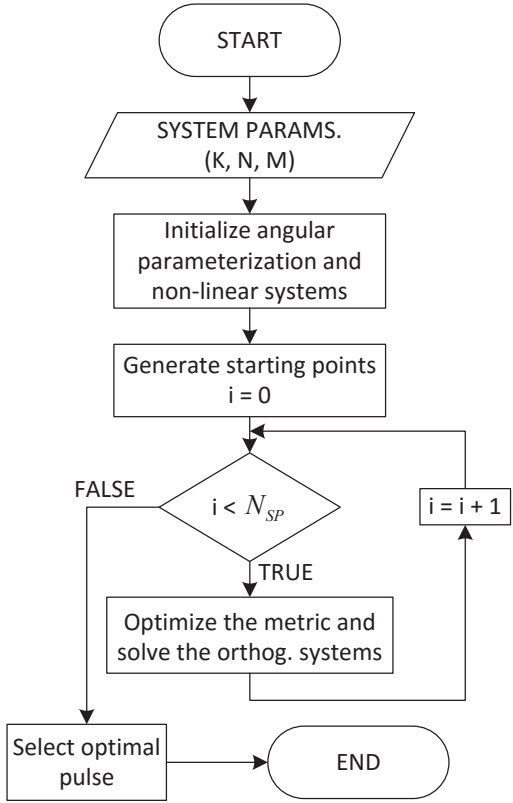}
   \caption{Flowchart of the design procedure.}
   \label{fig:Algorithm}
\end{figure}
A numerical approach has been followed to search for optimal orthogonal pulses that maximize the objective functions \eqref{eq:f1} and \eqref{eq:f2}. The optimization is a function maximization problem under non-linear constraints given by the orthogonality conditions \eqref{eq:ort_FD_1} and \eqref{eq:ort_FD_2}. The design is performed in the frequency domain and the angles representation, described in Sec. \ref{Sec:Design:Angles}, is adopted. The optimization process is performed exploiting the interior point method \cite{interiorMATLAB1,interiorMATLAB2,interiorMATLAB3} and repeated several times for randomly selected starting points. In detail, we generate $N_\text{SP}$ starting points, each corresponding to certain pulse coefficients obtained by randomly drawing angles in \eqref{eq:angles3}-\eqref{eq:angles4}. Then, for each starting point the optimization process (using the interior point method) is performed so that we obtain a set of solutions. Then, we select the solution in this set for which the benefit function (metric) is maximum. In this work, $N_\text{SP} = 500$.  The procedure is sketched in the flowchart of Fig.~\ref{fig:Algorithm}.

The design considers band limited pulses with three different $K/N$ ratio values: the critically sampled case ($K=N$), the maximum rate over-sampled case ($K, N=K+1$) and an over-sampled case with $K/N=2/3$. In detail, we have chosen $K \in \{8,10,12\}$ and the corresponding $N$ for  the same filter length s.t. $M > 300$. 

For the IBOB energy ratio metric, the considered pulses are real and even, so that the DFT coefficients are also real and the angles representation in \eqref{eq:angles3}-\eqref{eq:angles4} uses null phases. The design of a complex and symmetric pulse does not yield any improvement. 

For the capacity metric, we consider real and complex pulse FD responses. The former pulses have real and even DFT coefficients and consequently, real and even impulse responses. The latter pulses have complex DFT coefficients with Hermitian symmetry and, in general, no symmetry in time domain. Numerical results show that the complex solution outperforms the real pulse solution. 

The performance of the optimal pulse is compared with the one offered by a root-raised-cosine (RRC) baseline pulse. The pulse DFT coefficients are obtained sampling the frequency response of a RRC filter  \cite{CBFMT_ICC}. The roll-off is chosen so that there is no interference between adjacent sub-channels. In detail, we have $\beta_{max} = (Q-L)/L$. In the critically sampled case,  $\beta_{max}=0$. Thus, the RRC pulse coincides with the rectangular window in frequency domain. 

\renewcommand{\arraystretch}{1}
\begin{table*}[tb]
  \caption{IBOB energy ratios and maximum achievable rate for the baseline and the optimal pulses.}
  \label{tab:pratio_ort}
  \centering
\resizebox{4.75in}{!}{
  \begin{tabular}{|c|c|c|ccc|ccc|}
  \hline
  \multicolumn{3}{ |c| }{\multirow{3}{*}{\parbox{1.5cm}{\centering System parameters}}} & \multicolumn{6}{ c| }{Metrics values} \\ \cline{4-6} \cline{5-9}
  \multicolumn{3}{|c|  }{} &  \multicolumn{3}{ c| }{IBOB energy ratio [dB]}  & \multicolumn{3}{ c| }{Achievable rate [Mbps]} \\ \cline{4-9}
  \multicolumn{3}{|c|  }{} &  & Max & Max & & Max & Max \\ \cline{1-3}
  K & N & M & RRC & en. ratio pulse & capacity pulse & RRC & en. ratio pulse & capacity pulse \\ \hline \hline
  \multirow{3}{*}{8} & 8 & \multirow{3}{*}{360} & 20.62 & 20.62 & 18.02 & 96.57 & 96.57 & 96.57\\
  & 9 & & 45.33 & 102.17 & 40.68 & 98.93 & 97.89 & 102.50\\
  & 12 & & 56.88 & 127.11 & 43.89 & 92.21 & 74.27 & 100.44\\ \hline
\multirow{3}{*}{10} & 10 & \multirow{3}{*}{330} & 19.24 & 19.24 & 16.62 & 101.52 & 101.52 & 103.07\\
  & 11 & & 34.15 & 56.79 & 42.44 & 100.11 & 84.47 & 110.70\\
  & 15 & & 52.59 & 120.39 & 41.77 & 100.30 & 79.29 & 107.71\\ \hline
  \multirow{3}{*}{12} & 12 & \multirow{3}{*}{468} & 19.98 & 19.98 & 17.36 & 92.55 & 92.55 & 92.58\\
  & 13 & & 34.79 & 58.00 & 38.27 & 90.61 & 74.98 & 104.44\\
  & 18 & & 54.94 & 114.79 & 45.65 & 96.75 & 69.91 & 105.24\\ \hline
  \end{tabular}
  }
\end{table*}

\subsection{Maximum In-Band to Out-Band Energy Ratio} \label{Sec:Results:Ratio}
In Tab. \ref{tab:pratio_ort}, we summarize the results. The IBOB energy ratio achieved with the optimal pulse is shown for every considered set of parameters. The results show that the optimal pulse significantly improves the IBOB energy ratio w.r.t. the RRC pulse. In the critically sampled case, the RRC pulse is a rectangular pulse in frequency domain and it is the optimal solution. In the over-sampled case, the metric is improved w.r.t. the critically sampled case because $Q > L$ introduces some redundancy adding some degrees of freedom in the prototype pulse design. 

As an example, the frequency response for the $K=8, N=12$ case is shown in Fig. \ref{fig:pulse_example}. In the top-left plot we report the DFT of the pulse while the DTFT is shown in the bottom-left plot. Similar plots can be obtained for the other parameters. The pulses designed according to the IBOB energy ratio metric exhibit much higher spectrum confinement than the RRC pulse. The capacity performance, shown in the right plot of Fig. \ref{fig:pulse_example}, is discussed in the next section. 

 \begin{figure}[t]
\centering
\includegraphics[width=\figureWidthFigure]{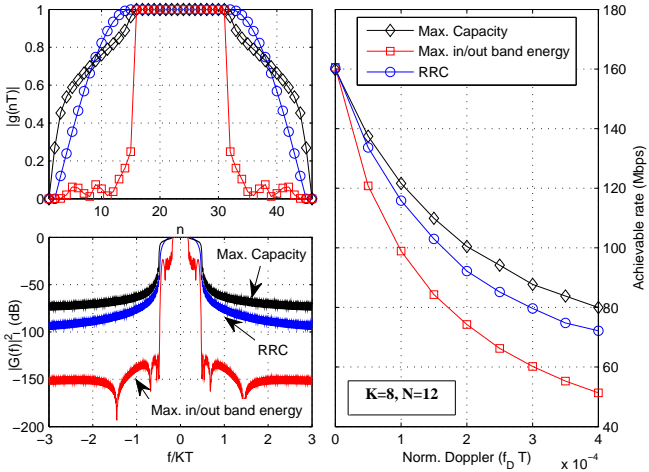}
   \caption{Prototype pulse example for $K=8, N=12, M=360$. 
   On top-left, the plot shows the amplitude of the prototype pulses DFT. On bottom-left, the frequency response of the pulses is shown. On right, the plot shows the maximum achievable rate as a function of the normalized Doppler frequency.}
   \label{fig:pulse_example}
\end{figure} 
\subsection{Maximum Achievable Rate} \label{Sec:Results:Capacity}
To design the prototype pulse that maximizes the achievable rate, we follow the steps reported in Sec. \ref{Sec:Design}. A time-variant and time-dispersive fading wireless channel with the Clarke's isotropic scattering model is assumed \cite{ref:Stuber_MobileComm}. Therefore, the channel coefficients are assumed to be independent zero-mean complex Gaussian random variables with correlation 
\begin{equation}
E\left[ \alpha_\ell^*(m) \alpha_{\ell'}(m+n) \right] = \Omega_{\ell} J_0(2 \pi f_D n )\delta(\ell-\ell'), \label{eq:clarke}
\end{equation}
where $f_D$ and $J_0(\cdot)$ are the maximum Doppler frequency and the zero-order Bessel function of the first kind, respectively. The power decay profile is exponential, i.e., $\Omega_{\ell} = \Omega_0\;e^{-\ell/\gamma }$, where $\Omega_0$ is a normalization constant to obtain unit average power, and $\gamma$ is the normalized, w.r.t. the sampling period, delay spread. The normalized delay spread is equal to $\gamma=2$ yielding a channel of duration $5$ samples. The CP has length $8$ coefficients and the sampling frequency is $1/T = 20 \, \text{MHz}$. The optimal capacity pulses have been designed assuming the normalized Doppler frequency equal to $2 \times 10^{-4}$. In all cases, the SNR is set equal to $40$ dB. A 1-tap MMSE equalizer is adopted, as described in Sec.~\ref{Sec:Practical:TV}. 

In Tab. \ref{tab:pratio_ort}, the average capacity is reported for the normalized Doppler frequency equal to $2 \times 10^{-4}$. In the critically sampled case, numerical results show that the rectangular window is the optimal solution for the capacity metric too. Similarly, to the IBOB energy ratio metric optimization case, no gains have been found by considering complex valued rectangular windows.

In the right plot of Fig. \ref{fig:pulse_example}, the average capacity as a function of the normalized Doppler frequency is shown for $K=8, N=12$. The optimal pulse for the IBOB energy ratio performs worse than the RRC pulse. Vice versa, the optimal capacity pulse increases performance. It should be noted that the spectrum of the optimal capacity wise pulse is less confined than that of the RRC pulse as the left plots of Fig. \ref{fig:pulse_example} show.  

The achievable rate is related to the SINR and the ratio $K/N$. The SINR increases when the rate $K/N$ decreases. On the contrary, as $K/N$ decreases the transmission rate decreases. To improve the achievable rate a trade-off is necessary. As an example, Fig. \ref{fig:Capacity-doppler} shows the achievable rate for $K=10$ and three different $K/N$ rates ($1$, $10/11$ and $2/3$). For the static channel (no Doppler), the RRC and the optimal pulse exhibit the same performance. When Doppler increases, the achievable rate decreases for all pulses due to the increase in interference. However, the optimal pulses show all higher capacity than the RRC pulse. Depending on the Doppler value, an optimal $K/N$ ratio can be identified. For normalized Doppler values below $0.5 \times 10^{-4}$, the CS solution is the best. For normalized Doppler frequencies between $0.5 \times 10^{-4}$ and $2 \times 10^{-4}$ the $K=10, N=11$ case offers the highest capacity, while for higher values of Doppler the $K=10, N=15$ case is the best. 

\renewcommand{\arraystretch}{1}
\begin{table*}[t]
  \caption{IBOB energy ratios and maximum achievable rate for the extended  pulses.}
  \label{tab:pratio_ort_ext}
  \centering
\resizebox{4.75in}{!}{
  \begin{tabular}{|c|c|c|ccc|ccc|}
  \hline
  \multicolumn{3}{ |c| }{\multirow{3}{*}{\parbox{1.5cm}{\centering Extended parameters ($\alpha_1 = 3$)}}} & \multicolumn{6}{ c| }{Metrics values} \\ \cline{4-6} \cline{5-9}
  \multicolumn{3}{|c|  }{} &  \multicolumn{3}{ c| }{IBOB energy ratio [dB]}  & \multicolumn{3}{ c| }{Achievable rate [Mbps]} \\ \cline{4-9}
  \multicolumn{3}{|c|  }{} &  & Max & Max & & Max & Max \\ \cline{1-3}
  K & N & M & RRC & en. ratio pulse & capacity pulse & RRC & en. ratio pulse & capacity pulse \\ \hline \hline
  \multirow{2}{*}{24} & 24 & \multirow{2}{*}{1080} & 20.62 & 20.62 & 18.01 & 77.93 & 77.93 & 78.61\\
  & 36 & & 59.93 & 130.00 & 46.74 & 101.01 & 54.30 & 108.70\\ \hline
\multirow{2}{*}{30} & 30 & \multirow{2}{*}{990} & 19.23 & 19.23 & 16.61 & 76.75 & 76.75 & 77.19\\
  & 45 & & 55.64 & 123.38 & 44.54 & 102.25 & 54.86 & 110.40\\ \hline
  \multirow{2}{*}{36} & 36 & \multirow{2}{*}{1404} & 19.98 & 19.98 & 17.41 & 78.12 & 78.12 & 78.54\\
  & 54 & & 57.97 & 117.81 & 48.71 & 105.07 & 50.14 & 108.38\\ \hline
  \end{tabular}
  }
\end{table*}

\begin{figure}[t]
\centering
\includegraphics[width=\figureWidthFigure]{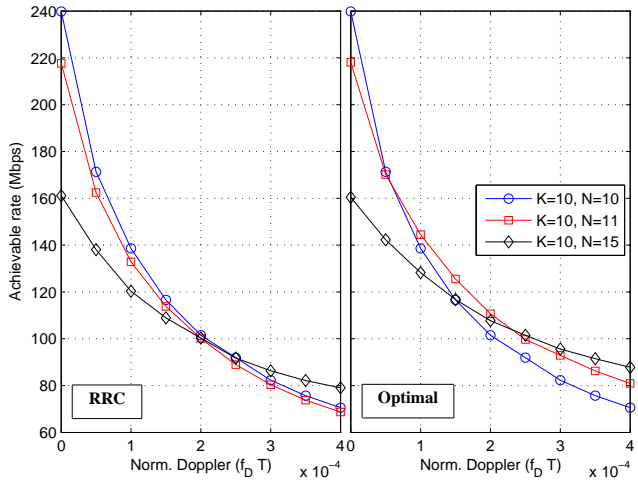}
   \caption{Achievable rate as a function of the normalized Doppler frequency for $K=10$ and $N \in \{10,11,15\}$. On the left, the achievable rate of the RRC pulse. On the right, the performances of the optimal pulse.}
   \label{fig:Capacity-doppler}
\end{figure}

\subsection{Performances with Extended Pulses} \label{Sec:Results:Capacity:extension}
In Sec.~\ref{Sec:Orthogonality:Params_variation}, we have discussed how to build an orthogonal pulse from a mother pulse designed for a set of parameters $(K, N, M)$ to a new set of parameters $(\alpha_1 K, \alpha_1 N,\alpha_1 M)$ or $(\alpha_2 K, \alpha_2 N, M)$.

As an example, Fig.~\ref{fig:Capacity-subopt} shows the achievable rate as a function of the normalized Doppler frequency for the pulses with parameters $(24,36,360)$ and $(24,36,1080)$ obtained from the mother optimal pulse designed for the parameters $(8,12,360)$. On the left plot, the prototype pulse is obtained by sampling in frequency domain by a factor $3$. On the right plot, the prototype pulse is obtained by zero padding the spectrum: the optimal pulse is extended to a pulse with $3 M$ coefficients.

In both cases, the new prototype pulses are orthogonal and exhibit performance better than the RRC pulse. Indeed, to obtain a capacity optimal pulse a search has to be conducted for every new set of parameters. However, the figure shows that these sub-optimal pulses offer a system capacity that is close to the optimal ones.  

In Tab.~\ref{tab:pratio_ort_ext}, the IBOB energy ratio and the average capacity are reported for the extended pulses. The extension described in Thm.~\ref{thm:params_1} is used, with $\alpha_1=3$.

\begin{figure}[t]
\centering
\includegraphics[width=\figureWidthFigure]{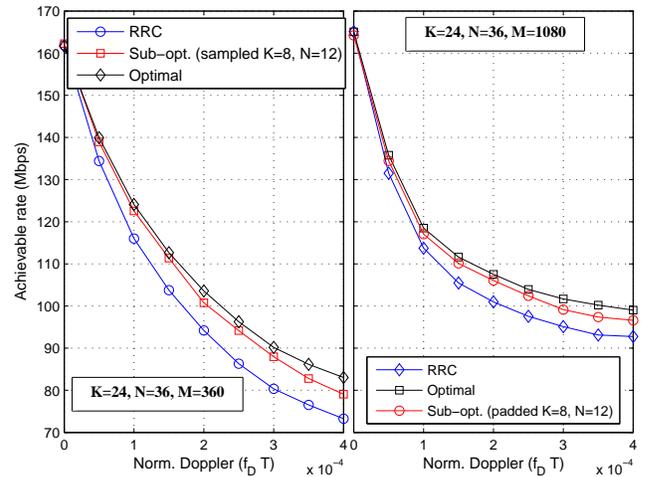}
   \caption{Achievable rate as a function of the normalized Doppler frequency for $K=24, N=36$. The orthogonal pulses are derived following the theorems reported in Sec. \ref{Sec:Orthogonality:Params_variation}. On the left, the sub-optimal pulse is derived by FD sampling the $K=8, N=12$ optimal pulse (constant filter length). On the right, the sub-optimal pulse is derived with a zero padding of the $K=8, N=12$ optimal pulse (filter length variation).}
   \label{fig:Capacity-subopt}
\end{figure}

\subsection{Remarks} \label{Sec:Results:Capacity:remarks}
Typically, the IBOB energy ratio metric is used in the filter bank design literature. This is because the frequency confinement of the sub-channels is the design criterion. The capacity metric is less, or even not at all, used. The latter is more appropriate when the goal is the design of a communication scheme with high spectral efficiency. The numerical results (see Tab.~\ref{tab:pratio_ort} and Tab.~\ref{tab:pratio_ort_ext}). allow also to compare the capacity offered by the filter bank designed with the two metrics, and of course, the latter metric offers always higher capacity.



\section{Conclusions}
\label{Sec:Conclusions}
In this paper, the orthogonal design of CB-FMT has been analyzed. Both time domain and frequency domain design criteria have been reported. It has been shown that the orthogonality conditions can be written in matrix form and they translate in a set of equations that the prototype pulse DFT coefficients must fulfill. A subset of these equations are redundant and can be removed allowing to use a system with a reduced number of unique equations. The search of optimal pulses has then been conducted exploiting a parameterization with hyper-spherical coordinates under an objective maximization function: either the IBOB energy ratio or the maximum system capacity. Numerical examples have been reported and show that better performance than the RRC pulse can be obtained in static and time variant frequency selective fading channels. Finally, it has been shown that from a mother pulse designed for a set of parameters $(K$, $N$, $M)$ it is simple to obtain an orthogonal solution with other sets of parameters, e.g. for an increase in the number of sub-channels with increased or constant pulse length. These pulses still offer better capacity performance than the RRC pulse in time-variant fading channels and it is close to the optimally designed pulses.

\appendices
\section{Proofs of Theorem \ref{thm:Matrix} and Corollary \ref{cor:critically}}
\label{Appx:MatrixForm}

\subsection{Proof of the Theorem}
\label{Appx:MatrixForm:Thm}

\subsubsection{\textbf{$\hat{\mathbf{H}}_p$ properties}}
\label{Appx:MatrixForm:Thm:Properties}
Before reporting the theorem proof, two proprieties of the $\hat{\mathbf{H}}_p$ matrices are introduced.
\let\oldlabelenumi=\labelenumi
\renewcommand{\labelenumi}{(\alph{enumi})}
\begin{enumerate}
 \item Each matrix $\hat{\mathbf{H}}_p$ involves $L/\gcd(Q,L)$ distinct vectors.
 
 A vector $\mathbf{v}_{p_2}$ belongs to $\hat{\mathbf{H}}_p$ if the following relation (derived from \eqref{eq:cpk}) is satisfied:
\begin{align}
p_2 &= p + kQ + aL, \label{eq:diofand1}\\
p_2  \in & \{0, \dots, L-1\},  a \in \Z. \notag
\end{align}
Equation \eqref{eq:diofand1} can be rewritten as
\begin{equation}
kQ+aL = \Delta_p, \label{eq:diofand2}
\end{equation}
where $\Delta_p = p_2-p$. Equation \eqref{eq:diofand2} is a linear diophantine equation and it has integer solution if $\Delta_p$ is a multiple of $\gcd(Q,L)$ \cite{mordell1969diophantine}. Thus, the matrix $\hat{\mathbf{H}}_p$ contains only $L/\gcd(Q,L)$ distinct vectors, i.e., the ones identified by the indexes
\begin{equation}
p_2 = p + b \gcd(Q,L), \label{eq:diofand3}
\end{equation}
where $b \in \Z$ is chosen s.t. $p_2 \in \{0, \dots, L-1 \}$.
 
 \item Given $N_s = \gcd(Q,L)$, the $\hat{\mathbf{H}}_p$ matrices with index $p \in \{ N_s, \dots, L-1\}$ are obtained by a swap of the rows and columns of the $\hat{\mathbf{H}}_p$ matrices with index $p \in \{0, \dots, N_s-1\}$. 
 
We start from the index relations
\begin{align}
c_{(p+(k_2 Q)_L,k)} &= (p+(k_2)_L(Q)_L+(k)_L(Q)_L)_L \notag\\
& = (p+(k+k_2)_L (Q)_L)_L \notag\\
&= c_{(p,k+k_2)}, \label{eq:cpk2}\\
d_{(p+(k_2 Q)_L,k)} &= \frac{p+kQ+(k_2 Q)_L- \left(p+\left(k+k_2\right) Q \right)_L}{L} \notag\\
&= d_{(p,k+k_2)} + \frac{(k_2Q)_L- k_2 Q}{L}, \label{eq:dpk2}
\end{align}
and we apply them to the matrices $\mathbf{H}_{\text{ort},p}$ and $\mathbf{H}_{\text{ort},p+(k_2 Q)_L}$. We can deduce the following properties:
\begin{itemize}
\item the two matrices are made of the same column vectors. The $k$-th column vector in $\mathbf{H}_{\text{ort},p}$ corresponds to the $(k+k_2)_L$-th column vector of the matrix $\mathbf{H}_{\text{ort},p+(k_2 Q)_L}$ (deduced from \eqref{eq:cpk2});
\item the column vectors of the matrix $\mathbf{H}_{\text{ort},p}$ are cyclically shifted by a factor $\left[ (k_2Q)_L - k_2 Q \right]/L$ in the matrix $\mathbf{H}_{\text{ort},p+(k_2 Q)_L}$ (deduced from \eqref{eq:dpk2});
\item the term $p+(k_2 Q)_L$ can be rewritten as in \eqref{eq:diofand3}. 
\end{itemize}
Thus, the matrices $\mathbf{H}_{\text{ort},p+b \gcd(Q, L)}, b \neq 0$ are redundant because they are obtained by a swap of the rows and columns of the matrix $\mathbf{H}_{\text{ort},p}$ \label{Hp_properties_2}
 \end{enumerate} 
\let\labelenumi=\oldlabelenumi

\subsubsection{\textbf{Proof}}
Given a matrix \eqref{eq:Hort2} for a certain index $p$, if all the columns are orthonormal, the inner product \eqref{eq:inner_product} is zero for every $k$. Thus,  \eqref{eq:ort_FD_2} is satisfied. 

Now, we focus on \eqref{eq:ort_FD_1}. This equation can be seen as the square Euclidean norm of the vector $\tau^{d_{(p,k)}} \left\{ \mathbf{v}_{c_{(p,k)}} \right\}$. If all the columns are orthonormal, so that the norm of the matrix \eqref{eq:Hort} is equal to $N$, \eqref{eq:ort_FD_1} will be satisfied. 

Finally --- exploiting the properties (b)reported in the previous section, it follows that the orthogonality conditions have to be fulfilled by only the matrices of index $p \in \{0, \dots, N_s-1\}$ because the others matrices are redundant.

\subsection{Proof of Corollary \ref{cor:critically}}
\label{Appx:MatrixForm:Corollary}
When the system is critically sampled, $Q=L$ so that \eqref{eq:cpk} and \eqref{eq:dpk} can be rewritten as 
\begin{align}
c_{(p,k)} &= (p+kL)_L = p, \label{eq:cpk3}\\
d_{(p,k)} &= \frac{p+kL-(p+kL)_L}{L} = k. \label{eq:dpk3}
\end{align}
Equation \eqref{eq:cpk3} shows that all the columns of $\mathbf{H}_{\text{ort},p}$ are made with the elements of the vector $\mathbf{v}_p$. Equation \eqref{eq:dpk3} shows that the $k$-th column is cyclically shifted by a factor $k$. Thus, the matrix is circulant.

A property of an orthogonal matrix is that all its eigenvalues, denoted as $\lambda_n$, are equal, in modulus, to $1$ \cite{golub1996matrix}. Furthermore, the $n$-th eigenvalue of a circulant matrix is obtained as \cite{davis1994circulant}
\begin{align}
\lambda_n = \sum_{k=0}^{N-1} \left\{ \mathbf{v}_p \right\}_k e^{-j 2 \pi nk/N}. \label{eq:eigen_circulant}
\end{align}
Equation \eqref{eq:eigen_circulant} can be seen as an $N$-point DFT. Thus, the eigenvalue vector $\Lambda = \left[\lambda_0, \dots, \lambda_{N-1} \right]$ is given by
\begin{equation}
\Lambda = \mathbf{F}_N \mathbf{v}_p.
\end{equation}  

\section{Proof of Theorem \ref{thm:geq}}
\label{Appx:geq}
Under the theorem hypothesis, replacing \eqref{eq:gk2} into \eqref{eq:CBFMT_TX} allows us to rewrite the transmitted signal as
\begin{equation}
 x(n) = \sum_{k=0}^{K-1} \sum_{\ell=0}^{L-1} a^{(k)}(\ell N) g^{(k)}(n-\ell N). \label{eq:CBFMT_TX_2}
 \end{equation} 
The block of coefficients at the receiver input in \eqref{eq:geq2}, after CP discard, is given by
\begin{align}
y(n) &= \sum_{s=0}^{P-1} g_\text{eq}(s) x(n-s) \notag \\
&=  \sum_{s=0}^{P-1} g_\text{eq}(s) \sum_{k=0}^{K-1} \sum_{\ell=0}^{L-1} a^{(k)}(\ell N) g^{(k)}(n-s-\ell N) \notag \\
&= \sum_{k=0}^{K-1} \sum_{\ell=0}^{L-1} a^{(k)}(\ell N) g_1^{(k)}(n-\ell N),\\
n \in & \{0,\dots, M-1\} \notag
\end{align}
where $g_1^{(k)}(n) = g^{(k)} \otimes  g_\text{eq}(n)$. Thus, the matched receiver filter for the $i$-th sub-channel is given by $h^{(i)}(n)= \left( g_\text{eq,-} \otimes g^{(i)}_-(n) \right)^*$. Condition \eqref{eq:rgh_TD} becomes
\begin{equation}
r_\text{eq}^{(k,i)}(mN) = r^{(k,i)} \otimes g_\text{eq}  \otimes g_\text{eq,-}^*(mN). \label{eq:rgheq_TD}
\end{equation}
The prototype pulse $g(n)$ is designed to be orthogonal and $r^{(k,i)}(nN)$ is equal to the Kronecker delta. Thus, condition \eqref{eq:rgheq_TD} simply becomes
\begin{equation}
r_\text{eq}^{(i,i)}(mN) = g_\text{eq} \otimes g_\text{eq,-}^*(mN). \label{eq:rgheq_TD_2}
\end{equation}
Finally, if \eqref{eq:rgheq_TD_2} is equal to $\delta_m$, i.e., the equivalent filter impulse response is orthogonal to its cyclic shifts (with period $M$) of multiples of $N$, then orthogonality will not be lost.

\section{Proof of Theorem \ref{thm:inf}}
\label{Appx:inf}
First, we consider the critically sampled case $K=N$. Corollary \ref{cor:critically} shows that every vector $\mathbf{v}_p$ s.t. 
\begin{equation}
\left| \left\{ \mathbf{F}_N \mathbf{v}_p \right\}_i \right| = 1 \label{eq:critically}
\end{equation}
is a valid solution. Equation \eqref{eq:critically} can be rewritten as
\begin{equation}
\left\{ \mathbf{F}_N \mathbf{v}_p \right\}_i  = e^{j \phi_i}. \label{eq:critically2}
\end{equation}
Thus, every set $\{\phi_0,\dots,\phi_{N-1}\}$ gives a valid solution and, finally, the number of solutions is infinite.

Now, we focus on the general case $K < N$. The matrix columns of $\mathbf{H}_{\text{ort},p}$ are composed of $L/N_s$ distinct vectors and theirs cyclically shifted version by a factor $LQ/N_s$. In detail, from \eqref{eq:cpk} and \eqref{eq:dpk}, we can see that
\begin{align}
c_{\left(p,k+b L/N_s \right)} &= \left(p+kQ+ L \frac{Q}{\gcd(Q,L)} \right) = c_{(p,k)},\\
d_{(p,k+b L/N_s)} &= \frac{p+kQ+ b LQ/N_s - c_{\left(p,k+L/N_s \right)} }{L} \notag\\
&= d_{(p,k)} + b \frac{LQ}{N_s},  
\end{align}
where $b \in \Z$ is chosen s.t. $k+b L/N_s \in \{0, \dots, K-1 \}$.

Given a $N \times N$ circulant orthogonal matrix, obtained from \eqref{eq:critically2}, we can obtain a valid $\mathbf{H}_{\text{ort},p}$ matrix by simply dropping $N-K$ columns. Given that there are infinite circulant orthogonal matrices, there are infinite solutions in the case $K<N$ too. 

\section{Proof of Theorem \ref{thm:params_1}}
\label{Appx:params_1}
To prove the orthogonality, we start from the no-ICI condition in \eqref{eq:ort_FD_3} rewritten as
\begin{align}
 \frac{1}{\alpha_1 N} \sum_{s=0}^{\alpha_1 N-1} G_{\alpha_1}(p+sL)& G_{\alpha_1}^*(p+sL+kQ) =0, \label{eq:ort_proof_ici_alpha}\\
      \forall p  \in \{0, \dots,L-1\},& \quad \forall k \in \{1,\dots,\alpha K-1\}. \notag
\end{align}
From \eqref{eq:ort_proof_ici_alpha}, $G_{\alpha_1}^*(p+sL+kQ)$ defined in \eqref{eq:Galpha_1} is null $\forall k \in \{1,\dots,\alpha_1 K-1\}$ due to the frequency confinement. Thus, the criterion is satisfied. We now focus on the no-ISI condition in \eqref{eq:ort_FD_1} rewritten as
\begin{align}
  \frac{1}{\alpha_1 N} \sum_{s=0}^{\alpha_1 N-1} & \left| G_{\alpha_1}(p+sL) \right|^2 = 1, \label{eq:ort_proof_isi_alpha_1_1}\\
  \forall p  \in & \{0, \dots,L-1\}. \notag
\end{align}
Eq. \eqref{eq:ort_proof_isi_alpha_1_1} can be rewritten as
\begin{align}
  \frac{1}{\alpha_1 N} \sum_{s=0}^{N-1} \left| G_{\alpha_1}(p+sL) \right|^2 & + \frac{1}{\alpha_1 N}  \sum_{s=N}^{\alpha_1 N-1} \left| G_{\alpha_1}(p+sL) \right|^2 \label{eq:ort_proof_isi_alpha_1_2} \\
 =  \frac{1}{\alpha_1 N} \sum_{s=0}^{N-1} \left| G_{\alpha_1}(p+sL) \right|^2 & = \frac{1}{N} \sum_{s=0}^{N-1} \left| G(p+sL) \right|^2 = 1 \label{eq:ort_proof_isi_alpha_1_3} \\ 
     \forall p  \in & \{0, \dots,L-1\}. \notag
\end{align}
In \eqref{eq:ort_proof_isi_alpha_1_2}, the second summation term is null, while \eqref{eq:ort_proof_isi_alpha_1_3} shows that the no-ISI condition is fulfilled. 

\section{Proof of Theorem \ref{thm:params_2}}
\label{Appx:params_2}
To prove the orthogonality, we start from the no-ICI condition in \eqref{eq:ort_FD_3}, rewritten as
\begin{align}
 \frac{1}{\alpha_2 N}  \sum_{s=0}^{\alpha_2 N-1} G_{\alpha_2}\left(p+s\frac{L}{\alpha_2}\right)& G_{\alpha_2}^*\left(p+s\frac{L}{\alpha_2}+k\frac{Q}{\alpha_2}\right) \label{eq:ort_proof_ici_alpha_2_1}\\
 =  \frac{1}{N} \sum_{s=0}^{\alpha_2 N-1} G(\alpha_2 p+sL)& G^*(\alpha_2 p+sL+kQ) = 0, \label{eq:ort_proof_ici_alpha_2_2}\\
      \forall p \in \left\{0, \dots,\frac{L}{\alpha_2}-1 \right\},& \quad \forall k \in \{1,\dots,\alpha_2 K-1\}. \notag
\end{align}
From \eqref{eq:ort_proof_ici_alpha_2_2}, the term $G^*(\alpha_2 p+sL+kQ)$ is null $\forall k \in \{1,\dots,\alpha_2 K-1\}$ due to the frequency confinement. Thus, the criterion is satisfied.

We now focus on the no-ISI condition in \eqref{eq:ort_FD_1}, rewritten as
\begin{align}
  \frac{1}{\alpha_2 N} \sum_{s=0}^{\alpha_2 N-1} & \left| G_{\alpha_2}\left(p+s \frac{L}{\alpha_2}\right) \right|^2, \label{eq:ort_proof_isi_alpha_2}\\
  \forall p  \in & \left\{0, \dots,\frac{L}{\alpha_2}-1 \right\}. \notag
\end{align}
Eq. \eqref{eq:ort_proof_isi_alpha_2} can be rewritten as
\begin{align}
  \frac{1}{\alpha_2 N} \sum_{s_1=0}^{\alpha_2-1}\sum_{s=s_1 N}^{(s_1+1) N-1} & \left| G_{\alpha_2}\left(p+s \frac{L}{\alpha_2}\right) \right|^2\\
  = \frac{1}{N} \sum_{s=0}^{N-1} & \left| G\left(\alpha_2 p+sL\right) \right|^2=1\\
  \forall p  \in & \left\{0, \dots,\frac{L}{\alpha_2}-1 \right\}. \notag
\end{align}
which proves the no-ISI condition.


\end{document}